\documentclass[10pt,conference,letterpaper]{IEEEtran}

\newtheorem{definition}{Definition}
\newtheorem{theorem}{Theorem}
\newtheorem{lemma}{Lemma}

\newtheorem{remark}{Remark}

\newtheorem{corollary}{Corollary}

\newcommand{\bydef} {{\buildrel{\triangle}\over =}}
\usepackage{amsmath}
\usepackage{amsfonts}
\usepackage{graphicx}
\usepackage{pstricks}
\usepackage{dsfont}
\begin{document}
\title{On the Capacity of Memoryless Finite-State Multiple Access Channels with Asymmetric Noisy State Information at the Encoders}
\author{\IEEEauthorblockN{Nevroz \c{S}en}
\IEEEauthorblockA{Queen's University\\
nsen@mast.queensu.ca}
\and
\IEEEauthorblockN{Giacomo Como}
\IEEEauthorblockA{M.I.T.\\
giacomo@mit.edu}
\and
\IEEEauthorblockN{Serdar Y\"{u}ksel}
\IEEEauthorblockA{
Queen's University\\
 yuksel@mast.queensu.ca}
\and
\IEEEauthorblockN{Fady Alajaji}
\IEEEauthorblockA{
Queen's University \\
fady@mast.queensu.ca}
}

\maketitle

\footnotetext{This work was supported in part by the Natural Sciences and Engineering Research Council of Canada (NSERC).}

\begin{abstract}
We consider the capacity of memoryless finite-state multiple access channel (FS-MAC) with causal asymmetric noisy state information available at both transmitters and complete state information available at the receiver. Single letter inner and outer bounds are provided for the capacity of such channels when the state process is independent and identically distributed. The outer bound is attained by observing that the proposed inner bound is tight for the sum-rate capacity.
\end{abstract}

\section{Introduction and Literature Review}
Modeling communication channels with a state process fits well for many physical scenarios. For single-user channels, the characterization of the capacity with various degrees of channel state information at the transmitter (CSIT) and at the receiver (CSIR) is well understood. Among them, Shannon \cite{Sha_Side} determined the capacity formula when causal noiseless state information is available at the transmitter, where state is independent identically distributed (i.i.d.). The same problem with non-causal side information is considered in \cite{Gelfand}. In \cite{Salehi}, Shannon's result is extended to the case where noisy state observation is available at both the transmitter and the receiver. Later, in \cite{Caire} this result has been shown to be a special case of Shannon's model and the authors also determined that when CSIT is a deterministic function of CSIR optimal codes can be constructed directly on the input alphabet.

In the multi-user setting, \cite{Das} provides a multi-letter characterization of the capacity region of time-varying multiple access channels (MACs) with various degrees of CSIT and CSIR. In \cite{Jafar}, a general framework for the capacity region of MACs with causal and non-causal CSI is presented. In a related work, MACs where the encoders have degraded information on the channel state, which is coded to the encoders, is considered \cite{Cemal}. In \cite{lap-ste-1}, memoryless FS-MACs with two independent states (see also \cite{lap-ste-2} for the single state case), each known causally and strictly causally to one encoder, are considered and an achievable rate region, which is shown to contain an achievable region where each user applies Shannon strategies, is proposed. In \cite{lap-ste-1} and \cite{lap-ste-2} it is also shown that strictly casual state information does not increase the sum-rate capacity. More recently, in \cite{basher} finite-state Markovian MACs with asymmetric delayed state information at the transmitters are studied and their capacity region is determined.

The most relevant work to our paper is \cite{GiacomoSerdar}, which obtained a single letter characterization of the capacity region for memoryless FS-MAC in which transmitters have asymmetric partial quantized state observations and the receiver has full state information. In this work, the authors were inspired from team decision theory \cite{Yuksel}, \cite{Wits}. We herein mainly adopt the converse technique presented in \cite{GiacomoSerdar} and partially extend it to a noisy setup. The present paper, thus, studies the FS-MAC in which each of the transmitters have an asymmetric state information which is corrupted by an i.i.d. noise process and the receiver has complete state information. We provide a single letter inner bound to the capacity region, in terms of Shannon strategies \cite{Sha_Side}. By observing that this inner bound is tight for the sum-rate capacity, we also provide an outer bound to the channel's capacity region. We modify the approach in \cite{GiacomoSerdar} to account for the fact that the decoder does not have access to the state information at the encoders, and that the past state information does not lead to a tractable recursion.

The rest of the paper is organized as follows. In Section \ref{full} we formally state the problem, present inner and outer bounds to the capacity region with the achievability and converse proofs and in Section \ref{conc} we present concluding remarks.

Throughout the paper we will use the following notations. A random variable will be denoted by an upper case letter $X$ and its particular realization by a lower case letter $x$. For a vector $v$, and a positive integer $i$, $v_{i}$ will denote the $i$-th entry of $v$, while $v_{[i]} = (v_1, \cdots, v_i)$ will denote the vector of the first $i$ entries of $v$. For a finite set $\mathcal{A}$, $\mathcal{P}(\mathcal{A})$ will denote the simplex of probability distributions over $\mathcal{A}$. Probability distributions are denoted by $P(\cdot)$ and subscripted by the name of the random variables and conditioning, e.g., $P_{U,T|V,S}(u,t|v,s)$ is the conditional probability of $(U=u,T=t)$ given $(V=v,S=s)$. Finally, for a positive integer $n$, we shall denote by $\mathcal{A}^{(n)} := \bigcup_{0<s<n} \mathcal{A}^{s}$ the set of $\mathcal{A}$-strings of length smaller than $n$. We denote the indicator function of an event by $1_{\{E\}}$. All sets considered hereafter are finite.

\section{On the Capacity of FS-MAC with Noisy CSIT and Complete CSIR}\label{full}
Consider a two-user memoryless FS-MAC, with two encoders, $a, b$, and two independent message sources $W_a$ and $W_b$ which are uniformly distributed in the sets $W_a \in \{1,2,\cdots,M_a\}$ and $W_b \in \{1,2,\cdots,M_b\}$, respectively. The channel inputs of the encoders are $X^{a}$ and $X^{b}$, respectively. The channel state process is modeled as a sequence $\{S_t\}_{t=1}^{\infty}$ of i.i.d. random variables in some space $\cal S$. The two encoders have access to causal noisy version of the state information $S_t$ at each time $t\geq1$ modeled by $S_{t}^{a} \in {\cal S}^{a}$, $S_{t}^{b} \in {\cal S}^{b}$, respectively and as such the joint distribution of $(S_t,S_t^a,S_t^b)$ satisfies
\begin{eqnarray}
P_{S_{t}^{a}, S_{t}^{b},S_t}(s_{t}^{a}, s_{t}^{b},s_t)=P_{S_{t}^{a}|S_t}(s_{t}^{a}|s_t)P_{S_{t}^{b}|S_t}(s_{t}^{b}|s_t)P_{S_t}(s_t)\label{eq-sta-no}.
\end{eqnarray}
We also assume that $S_t$ is fully available at the receiver (see \ref{fig:macfb}) and that $(S_t,S_t^a,S_t^b)$ are independent of $(W_a,W_b)$ $\forall t\geq 1$. The channel inputs at time $t$, i.e., $X_t^{a}$ and $X_t^{b}$, are functions of the locally available information $(W_a,S_{[t]}^{a})$ and $(W_b,S_{[t]}^{b})$. Let $\mathbf{W}:=(W_a,W_b)$ and $\mathbf{X}:=(X^a,X^b)$. Then, the laws governing $n$-sequences of state, input and output letters are given by
\begin{eqnarray}
&&\hspace{-0.6in}P_{Y_{[n]}|\mathbf{W},\mathbf{X}_{[n]},S_{[n]},S_{[n]}^a,S_{[n]}^b}(y_{[n]}|\mathbf{w},\mathbf{x}_{[n]},s_{[n]},s_{[n]}^a,s_{[n]}^b)\nonumber\\
&&\quad\quad\quad=\prod _{t=1}^nP_{Y_t|X_{t}^{a}, X_{t}^{b}, S_t}(y_t|x_{t}^{a}, x_{t}^{b}, s_t), \label{eq-ch}
\end{eqnarray}
where the channel's transition distribution, $P_{Y_t|X_{t}^{a}, X_{t}^{b}, S_t}(y_t|x_{t}^{a}, x_{t}^{b}, s_t)$, is given a priori.
\begin{definition}\label{def:maccode}
An $(n, 2^{nR_a}, 2^{nR_b})$ code with block length $n$ and rates $(R_a, R_b)$ for an FS-MAC with noisy state feedback consists of
\begin{itemize}
\item [(1)] A sequence of mappings for each encoder
\begin{center}
$\phi_{t}^{(a)}: (\mathcal{S}^{a})^t \times \mathcal{W}_a  \rightarrow {\cal X}_a, \enspace t=1,2,...n$;\\
\end{center}
\begin{center}
$\phi_{t}^{(b)}: (\mathcal{S}^{b})^t \times \mathcal{W}_b \rightarrow {\cal X}_b, \enspace t=1,2,...n$.\\
\end{center}
\item [2)] An associated decoding function
\begin{center}
$\psi: (\mathcal{S})^n\times{\cal Y}^n\rightarrow \mathcal{W}_a \times \mathcal{W}_b$.\\
\end{center}
\end{itemize}
\end{definition}
%
%
\begin{figure}
%
%
\hspace{-0.1in}
\setlength{\unitlength}{3947sp}%
\begingroup\makeatletter\ifx\SetFigFont\undefined%
\gdef\SetFigFont#1#2#3#4#5{%
  \reset@font\fontsize{#1}{#2pt}%
  \fontfamily{#3}\fontseries{#4}\fontshape{#5}%
  \selectfont}%
\fi\endgroup%
\begin{picture}(4299,2635)(1281,-3202)
\thinlines
{\color[rgb]{0,0,0}\put(1346,-1379){\vector( 1, 0){300}}
}%
{\color[rgb]{0,0,0}\put(1658,-1725){\framebox(846,709){}}
}%
{\color[rgb]{0,0,0}\put(1656,-2709){\framebox(846,709){}}
}%
{\color[rgb]{0,0,0}\put(1346,-2327){\vector( 1, 0){300}}
}%
{\color[rgb]{0,0,0}\put(2510,-2328){\line( 1, 0){305}}
\put(2813,-2328){\line( 0, 1){344}}
\put(2813,-1984){\vector( 1, 0){162}}
}%
{\color[rgb]{0,0,0}\put(2991,-2224){\framebox(1036,706){}}
}%
{\color[rgb]{0,0,0}\put(2510,-1325){\line( 1, 0){308}}
\put(2817,-1325){\line( 0,-1){403}}
\put(2817,-1728){\vector( 1, 0){173}}
}%
{\color[rgb]{0,0,0}\put(2075,-579){\vector( 0,-1){431}}
\put(2075,-579){\line( 1, 0){1379}}
\put(3454,-579){\line( 0,-1){931}}
}%
{\color[rgb]{0,0,0}\put(2040,-3190){\vector( 0, 1){479}}
\put(2040,-3190){\line( 1, 0){1415}}
\put(3455,-3190){\line( 0, 1){970}}
}%
{\color[rgb]{0,0,0}\put(4034,-1713){\vector( 1, 0){343}}
}%
{\color[rgb]{0,0,0}\put(4034,-2026){\vector( 1, 0){343}}
}%
{\color[rgb]{0,0,0}\put(4384,-2234){\framebox(846,709){}}
}%
{\color[rgb]{0,0,0}\put(5237,-1696){\vector( 1, 0){300}}
}%
{\color[rgb]{0,0,0}\put(5237,-2079){\vector( 1, 0){300}}
}%
\put(1460,-1327){\makebox(0,0)[b]{\smash{{\SetFigFont{9}{10.8}{\rmdefault}{\mddefault}{\updefault}{\color[rgb]{0,0,0}$W_a$}%
}}}}
\put(2090,-1534){\makebox(0,0)[b]{\smash{{\SetFigFont{9}{10.8}{\rmdefault}{\mddefault}{\updefault}{\color[rgb]{0,0,0}$\phi_{t}^a(W_a,S_t^a)$}%
}}}}
\put(2088,-1292){\makebox(0,0)[b]{\smash{{\SetFigFont{9}{10.8}{\rmdefault}{\mddefault}{\updefault}{\color[rgb]{0,0,0}Encoder}%
}}}}
\put(1460,-2264){\makebox(0,0)[b]{\smash{{\SetFigFont{9}{10.8}{\rmdefault}{\mddefault}{\updefault}{\color[rgb]{0,0,0}$W_b$}%
}}}}
\put(2088,-2306){\makebox(0,0)[b]{\smash{{\SetFigFont{9}{10.8}{\rmdefault}{\mddefault}{\updefault}{\color[rgb]{0,0,0}Encoder}%
}}}}
\put(2090,-2521){\makebox(0,0)[b]{\smash{{\SetFigFont{9}{10.8}{\rmdefault}{\mddefault}{\updefault}{\color[rgb]{0,0,0}$\phi_{t}^b(W_b,S_t^b)$}%
}}}}
\put(3510,-2043){\makebox(0,0)[b]{\smash{{\SetFigFont{9}{10.8}{\rmdefault}{\mddefault}{\updefault}{\color[rgb]{0,0,0}$P(Y_t|\mathbf{X}_t,S_t)$}%
}}}}
\put(3502,-1796){\makebox(0,0)[b]{\smash{{\SetFigFont{9}{10.8}{\rmdefault}{\mddefault}{\updefault}{\color[rgb]{0,0,0}Channel}%
}}}}
\put(2682,-1250){\makebox(0,0)[b]{\smash{{\SetFigFont{8}{9.6}{\rmdefault}{\mddefault}{\updefault}{\color[rgb]{0,0,0}$X_t^a$}%
}}}}
\put(2682,-2476){\makebox(0,0)[b]{\smash{{\SetFigFont{9}{10.8}{\rmdefault}{\mddefault}{\updefault}{\color[rgb]{0,0,0}$X_t^b$}%
}}}}
\put(2164,-2969){\makebox(0,0)[b]{\smash{{\SetFigFont{9}{10.8}{\rmdefault}{\mddefault}{\updefault}{\color[rgb]{0,0,0}$S_{[t]}^b$}%
}}}}
\put(2222,-875){\makebox(0,0)[b]{\smash{{\SetFigFont{9}{10.8}{\rmdefault}{\mddefault}{\updefault}{\color[rgb]{0,0,0}$S_{[t]}^a$}%
}}}}
\put(4210,-1674){\makebox(0,0)[b]{\smash{{\SetFigFont{9}{10.8}{\rmdefault}{\mddefault}{\updefault}{\color[rgb]{0,0,0}$Y_t$}%
}}}}
\put(4210,-1945){\makebox(0,0)[b]{\smash{{\SetFigFont{9}{10.8}{\rmdefault}{\mddefault}{\updefault}{\color[rgb]{0,0,0}$S_{[t]}$}%
}}}}
\put(4764,-1797){\makebox(0,0)[b]{\smash{{\SetFigFont{9}{10.8}{\rmdefault}{\mddefault}{\updefault}{\color[rgb]{0,0,0}Decoder}%
}}}}
\put(4801,-2063){\makebox(0,0)[b]{\smash{{\SetFigFont{9}{10.8}{\rmdefault}{\mddefault}{\updefault}{\color[rgb]{0,0,0}$\psi(Y_{[n]},S_{[n]})$}%
}}}}
\put(5363,-1662){\makebox(0,0)[b]{\smash{{\SetFigFont{9}{10.8}{\rmdefault}{\mddefault}{\updefault}{\color[rgb]{0,0,0}$\hat{W}_a$}%
}}}}
\put(5363,-2054){\makebox(0,0)[b]{\smash{{\SetFigFont{9}{10.8}{\rmdefault}{\mddefault}{\updefault}{\color[rgb]{0,0,0}$\hat{W}_b$}%
}}}}
\end{picture}%
\caption{The multiple-access channel with noisy state feedback.}
\label{fig:macfb}
\end{figure}
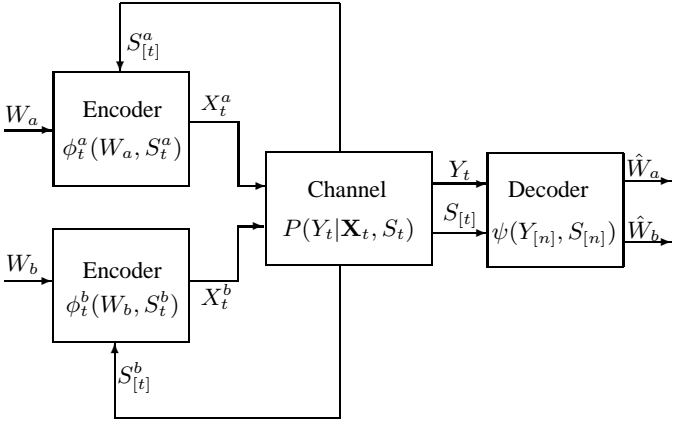
The system's probability of error, $P_{e}^{(n)}$, is given by
\begin{eqnarray}
\frac{1}{2^{n(R_a+R_b)}}\sum_{w_a=1}^{2^{nR_a}}\sum_{w_b=1}^{2^{nR_b}}P\left(\psi(Y_{[n]},S_{[n]})\neq (w_a,w_b)| \mathbf{W}=\mathbf{w}\right).\nonumber
\end{eqnarray}
A rate pair $(R_a, R_b)$ is achievable if for any $\epsilon > 0$ there exists, for all $n$ sufficiently large, an $(n, 2^{nR_a}, 2^{nR_b})$ code such that $\frac{1}{N}\log M_a \geq R_a \geq 0$, $\frac{1}{N}\log M_b \geq R_b \geq 0$ and $P_{e}^{(n)} \leq \epsilon$. The capacity region of the FS-MAC, ${\cal C}_{FS}$, is the closure of the set of all achievable rate pairs $(R_a, R_b)$ and the sum-rate capacity is defined as ${\cal C}^{FS}_{\sum}:=\max_{(R_a,R_b) \in {\cal C}_{FS}}(R_a+R_b)$.

Before proceeding with the main result, we introduce \textit{memoryless stationary team policies} \cite{GiacomoSerdar} and their associated rate regions. We first define Shannon strategies.
\begin{definition}
Let the set of all possible functions from ${\cal S}^{a}$ to ${\cal X}^{a}$ and ${\cal S}^{b}$ to ${\cal X}^{b}$  be denoted by ${\cal T}^{a}$ and ${\cal T}^{b}$, respectively, where ${\cal T}^{a}={{\cal X}^{a}}^{|{\cal S}^a|}$ and ${\cal T}^{b}={{\cal X}^{b}}^{|{\cal S}^b|}$. Let $T^{a} \in {\cal T}^{a}$ and $T^b \in {\cal T}^{b}$ be two ${\cal T}^a$-valued and ${\cal T}^b$-valued random vectors, respectively, referred to as Shannon strategies.
\end{definition}
\begin{definition}\cite{GiacomoSerdar}\label{def:teampol}
A memoryless stationary (in time) team policy is a family
\begin{eqnarray}
\Pi=\left\{\pi=\left(\pi_{T^a}(\cdot),\pi_{T^b}(\cdot)\right)\in {\cal P}({\cal T}^a)\times{\cal P}({\cal T}^b)\right\}\label{eq-team-pol}
\end{eqnarray}
of probability distributions on the two sets of random functions. For every memoryless stationary team policy $\pi$, $\mathcal{R}(\pi)$ denotes the region of all rate pairs $R=(R_a,R_b)$ satisfying
\begin{eqnarray}
R_a &<& I(T^{a};Y|T^{b},S) \quad \label{eq-ra1-f}\\
R_b &<& I(T^{b};Y|T^{a},S) \quad  \label{eq-ra2-f}\\
R_a+R_b &<& I(T^a,T^{b};Y|S) \label{eq-ra3-f}
\end{eqnarray}
where $S$, $T^a$, $T^b$ and $Y$ are random variables taking values in ${\cal S}$, ${\cal T}^{a}$, ${\cal T}^{b}$ and $\cal Y$, respectively and whose joint probability distribution factorizes as
\begin{eqnarray}
&&\hspace{-0.5in}P_{S,T^a,T^b,Y}(s,t^a,t^b,y) \nonumber\\
&&=P_S(s)P_{Y|T^a,T^b,S}(y|t^a,t^b,s)\pi_{T^a}(t^a)\pi_{T^b}(t^b)\label{eq-joidist-f}.
\end{eqnarray}
\end{definition}
We can now state the inner bound to the capacity region. Let ${\cal C}_{IN}:=\overline{co}\bigg(\bigcup_{\pi}\mathcal{R}(\pi)\bigg)$ denotes the closure of the convex hull of the rate regions $\mathcal{R}(\pi)$ given by (\ref{eq-ra1-f})-(\ref{eq-ra3-f}) associated to all possible memoryless stationary team polices as defined in (\ref{eq-team-pol}).
\begin{theorem}[Inner Bound]\label{the-main-achi-f}
 ${\cal C}_{IN}\subseteq {\cal C}_{FS}$.
\end{theorem}
The achievability proof follows the standard arguments of joint $\epsilon$-typical $n$-sequences \cite[Section 15.2]{cover}.
\begin{definition}\cite[Section 15.2]{cover}\label{def:jnt-typ}
The set ${\cal A}_{\epsilon}^{n}$ of $\epsilon$-typical $n$-sequences $\{(x^1_{[n]},\cdots,x^k_{[n]})\}$ with respect to the distribution $P_{X^1,\cdots,X^k}(x^1, \cdots, x^k)$ is defined by
\begin{eqnarray}
&&\hspace{-0.3in}{\cal A}_{\epsilon}^{n}=\left\{(x^1_{[n]},\cdots,x^k_{[n]})\in {\cal X}^1\times \cdots {\cal X}^k:\right.\nonumber\\
&&\hspace{-0.1in}\left.|-\frac{1}{n}\log\left(P(\textbf{s})\right)-H(S)|<\epsilon, \forall S \subseteq \{X^1,\cdots,X^k\}\right\} \nonumber
\end{eqnarray}
where $\textbf{s}$ denotes ordered set of sequences in $x^1_{[n]},\cdots,x^k_{[n]}$ corresponding to $S$.
\end{definition}
\begin{proof}[Proof of Theorem \ref{the-main-achi-f}]
Fix $(R_a,R_b)\in \mathcal{R}(\pi)$.

\textbf{\textit{Codebook Generation}}
Fix $\pi_{T^a}(t^a)$ and $\pi_{T^b}(t^b)$. For each $w_a \in \{1,\cdots, 2^{nR_a}\}$ randomly generate its corresponding $n$-tuple $t_{[n],w_a}^{a}$, each according to $\prod_{i=1}^n\pi_{T_{i}^{a}}(t_{i,w_a}^{a})$. Similarly, For each $w_b \in \{1,\cdots, 2^{nR_b}\}$ randomly generate its corresponding $n$-tuple $t_{[n],w_b}^{b}$, each according to $\prod_{i=1}^n\pi_{T_{i}^{b}}(t_{i,w_b}^{b})$. These codeword pairs form the codebook, which is revealed to the decoder while codewords $t_{i,w_l}^{l}$ is revealed to encoder $l$, $l=\{a,b\}$.

\textbf{\textit{Encoding}}
Define the encoding functions as follows: $x_{i}^{a}(w_a)=\phi_{i}^{a}(w_a,s_{[i]}^a)=t_{i,w_a}^a(s_i^a)$ and $x_{i}^{b}(w_b)=\phi_{i}^{b}(w_b,s_{[i]}^b)=t_{i,w_b}^b(s_i^b)$ where $t_{i,w_a}^a$ and $t_{i,w_b}^b$ denote the $i$th component of $t_{[n],w_a}^{a}$ and $t_{[n],w_b}^{b}$, respectively, and $s_i^a$ and $s_i^b$ denote the last component $s_{[i]}^a$ and $s_{[i]}^b$, respectively, $i=1,\cdots,n$.
Therefore, to send the messages $w_a$ and $w_b$, we simply transmit the corresponding $t_{[n],w_a}^{a}$ and $t_{[n],w_b}^{b}$, respectively.

\textbf{\textit{Decoding}}
After receiving $(y_{[n]},s_{[n]})$, the decoder looks for the only $(w_a,w_b)$ pair such that $(t_{[n],w_a}^{a}, t_{[n],w_b}^{b}$ $,y_{[n]},s_{[n]})$ are jointly $\epsilon-$typical and declares this pair as its estimate $(\hat{w}_a,\hat{w}_b)$.

\textbf{\textit{Error Analysis}}
Without loss of generality, we can assume that $(w_a,w_b)=(1,1)$ was sent. An error occurs, if the correct codewords are not typical with the received sequence or there is a pair of incorrect codewords that are typical with the received sequence. Define the events $E_{\alpha,\beta}\bydef\big\{(T_{[n],\alpha}^a,T_{[n],\beta}^b,Y_{[n]},S_{[n]})\in A_{\epsilon}^n\big\}$, $\alpha\in\{1,\cdots,2^{nR_a}\}$ and $\beta\in\{1,\cdots,2^{nR_b}\}$.
Then by the union bound we get
\begin{eqnarray}
&&\hspace{-0.3in}P_{e}^{n}=P\big(E_{1,1}^c\bigcup_{(\alpha,\beta)\neq(1,1)}E_{\alpha,\beta}\big)\leq P(E_{1,1}^c)\nonumber\\
&&\hspace{-0.3in}+\sum_{\alpha=1,\beta\neq1}P(E_{\alpha,\beta}) + \sum_{\alpha\neq1,\beta=1}P(E_{\alpha,\beta})+ \sum_{\alpha\neq1,\beta\neq1}P(E_{\alpha,\beta})\nonumber\\\label{eq-erbound}
\end{eqnarray}
where $P(E_{1,1}^c)$ denotes the probability that no message pair is jointly typical. It can easily be verified that $\{Y_i,S_i,T_i^a,T_i^b\}_{i=1}^{\infty}$ is an i.i.d. sequence and by \cite[Theorem 15.1.2]{cover}, $ P(E_{1,1}^c)\rightarrow0$ for sufficiently large $n$.
Next, let us consider the second term
\begin{eqnarray}
&&\hspace{-0.5in}\sum_{\alpha=1,\beta\neq1}P(E_{\alpha=1,\beta\neq1})\nonumber\\
&&\hspace{-0.2in}=\sum_{\alpha=1,\beta\neq1}P((T_{[n],1}^a,T_{[n],\beta}^b,Y_{[n]},S_{[n]})\in A_{\epsilon}^n)\nonumber\\
&&\hspace{-0.2in}\overset{(i)}{=}\sum_{\alpha=1,\beta\neq1}\sum_{(t_{[n]}^a,t_{[n]}^b,y_{[n]},s_{[n]})\in A_{\epsilon}^n}P_{T_{[n]}^b}(t_{[n]}^b)\nonumber\\
&&\quad \quad \quad P_{T_{[n]}^a, Y_{[n]},S_{[n]}}(t_{[n]}^a, y_{[n]},s_{[n]})\nonumber\\
&&\hspace{-0.2in}\overset{}{\leq}\sum_{\alpha=1,\beta\neq1}|A_{\epsilon}^n|2^{-n[H(T^b)-\epsilon]}2^{-n[H(T^a,Y,S)-\epsilon]}\nonumber\\
&&\hspace{-0.2in}\leq2^{nR_b}2^{-n[H(T^b)+ H(T^a,Y,S) - H(T^a,T^b,Y,S)-3\epsilon]}\nonumber\\
&&\hspace{-0.2in}\overset{(ii)}{=}2^{n[R_b-I(T^b;Y|S,T^a)-3\epsilon]}\label{eq-erbo1}
\end{eqnarray}
where $(i)$ holds since for $\beta\neq 1$, $T_{[n],\beta}^b$ is independent of $(T_{[n],1}^a,Y_{[n]},S_{[n]})$ and $(ii)$ follows since $T^b$ and $(T^a,S)$ are independent and $I(T^b;Y,T^a,S)=I(T^b;T^a,S)+I(T^b;Y|T^a,S)$
$=I(T^b;Y|T^a,S)$, where $I(T^b;T^a,S)=0$. Following the same steps for $(\alpha\neq1,\beta=1)$ and $(\alpha\neq1,\beta\neq1)$ we get
\begin{eqnarray}
&&\sum_{\alpha\neq1,\beta=1}P(E_{\alpha,\beta})\leq2^{n[R_a-I(T^a;Y|T^b,S)-3\epsilon]}, \nonumber\\ &&\sum_{\alpha\neq1,\beta\neq1}P(E_{\alpha,\beta})\leq2^{n[R_a+R_b-I(T^a,T^b;Y|S)-3\epsilon]} \label{eq-erbo2},
\end{eqnarray}
and the rate conditions of the $\mathcal{R}(\pi)$ imply that each term tends in (\ref{eq-erbound}) tends to zero as $n \rightarrow \infty$. This shows the achievability of a rate pair $(R_a, R_b) \in \mathcal{R}(\pi)$. Achievability of any rate pair in ${\cal C}_{IN}$ follows from a standard time-sharing argument.
\end{proof}
We now present an outer bound to ${\cal C}_{FS}$, which is obtained by providing a tight converse to the sum-rate capacity. Let \begin{eqnarray}
&&\hspace{-0.3in}{\cal C}_{OUT}:=\biggl\{(R_a,R_b)\in \mathds{R}^{+}\times\mathds{R}^{+}:\nonumber\\
&&\quad \quad \quad \quad \quad R_a+R_b\leq\sup_{\pi_{T^a}(t^a)\pi_{T^b}(t^b)}I(T^a,T^b;Y|S)\biggr\}\nonumber.
\end{eqnarray}
\begin{theorem}[Outer Bound]\label{the-main-outer-f}
${\cal C}_{FS}\subseteq {\cal C}_{OUT}$.
\end{theorem}
\begin{proof}
We need to show that all achievable rates satisfy
\begin{eqnarray}
R_a+R_b \leq \sup_{\pi_{T^a}(t^a)\pi_{T^b}(t^b)}I(T^a,T^b;Y|S),\nonumber
\end{eqnarray}
i.e., a converse for the sum-rate capacity. We use the converse technique of \cite{GiacomoSerdar} and extend it to a noisy setup. Therefore, following \cite{GiacomoSerdar} let $\alpha_{\mathbf{\sigma}}:= \frac{1}{n}P_{S_{[t-1]}}(\mathbf{\sigma})$, $\eta(\epsilon):=\frac{\epsilon}{1-\epsilon}\log|{\cal Y}|+\frac{H(\epsilon)}{1-\epsilon}$. Observe that $\lim_{\epsilon\rightarrow 0}\eta(\epsilon)=0$ and
\begin{eqnarray}
\sum_{\mathbf{\sigma} \in {{\cal S}}^{(n)}}\alpha_{\mathbf{\sigma}}=\frac{1}{n}\sum_{1\leq t \leq n}\enspace \sum_{\mathbf{\sigma} \in {{\cal S}}^{(t-1)}}P_{S_{[t-1]}}(\mathbf{\sigma})=1\nonumber,
\end{eqnarray}
where ${{\cal S}}^{(n)}$ and ${{\cal S}}^{(t-1)}$ are the sets of all ${\cal S}$-strings of length $n$ and $(t-1)$, respectively. First recall that, since $X_t^a=\phi_{t}^{(a)} \left(W_a,S_{[t-1]}^a,S_t^a\right)$ and $X_t^b=\phi_{t}^{(b)} \left(W_b,S_{[t-1]}^b, S_t^b\right)$, we have
\begin{eqnarray}
&&T_t^a=\phi_{t}^{(a)}\left(W_a,S_{[t-1]}^a\right) \in {{\cal X}^{a}}^{|{\cal S}^a|},\nonumber\\ &&T_t^b=\phi_{t}^{(b)}\left(W_b,S_{[t-1]}^b\right) \in {{\cal X}^{b}}^{|{\cal S}^b|}. \label{eq-t2}
\end{eqnarray}
We now show that the sum of any achievable rate pair can be written as the convex combination of conditional mutual information terms which are indexed by the realization of past complete state information.
\begin{lemma}\label{lem:conv-f}
Let $T_t^{a} \in {\cal T}^{a}$ and $T_t^b \in {\cal T}^{b}$ be the Shannon strategies induced by $\phi_{t}^a$ and $\phi_{t}^b$, respectively, as shown in (\ref{eq-t2}). Assume that a rate pair $R=(R_a,R_b)$, with block length $n\geq1$ and a constant $\epsilon \in (0,1/2)$, is achievable. Then,
\begin{eqnarray}
R_a+R_b\leq \sum_{\mathbf{\sigma} \in {\cal S}^{(n)}}\alpha_{\mathbf{\sigma}} I(T_t^a,T_t^b;Y_t|S_t,S_{[t-1]}=\mathbf{\sigma})+\eta(\epsilon) \label{eq-lf1}.
\end{eqnarray}
\end{lemma}
\begin{proof}
Let $\mathbf{T}_t:=(T_t^a,T_t^b)$. By Fano's inequality, we get
\begin{eqnarray}
H(\mathbf{W}|Y_{[n]},S_{[n]})\leq H(\epsilon) + \epsilon\log(|{\cal W}_a||{\cal W}_b|). \label{eq-c1}
\end{eqnarray}
Observing that
\begin{eqnarray}
\hspace{-0.2in}I(\mathbf{W};Y_{[n]}, S_{[n]})&=&H(\mathbf{W})-H(\mathbf{W}|Y_{[n]},S_{[n]})\nonumber\\
\hspace{-0.2in}&=&\log(|{\cal W}_a||{\cal W}_b|)-H(\mathbf{W}|Y_{[n]},S_{[n]}). \label{eq-c2}
\end{eqnarray}
Combining (\ref{eq-c1}) and (\ref{eq-c2}) gives
\begin{eqnarray}
(1-\epsilon)\log(|{\cal W}_a||{\cal W}_b|)\leq I(\mathbf{W};Y_{[n]}, S_{[n]})+H(\epsilon)\nonumber
\end{eqnarray}
and
\begin{eqnarray}
R_a+R_b&=& \frac{1}{n}\log(|{\cal W}_a||{\cal W}_b|)\nonumber\\
&\leq& \frac{1}{1-\epsilon}\frac{1}{n}\left(I(\mathbf{W};Y_{[n]}, S_{[n]})+H(\epsilon)\right)\label{eq-c3}.
\end{eqnarray}
Furthermore, $I(\mathbf{W};Y_{[n]}, S_{[n]})$ can be written as
\begin{eqnarray}
&&\hspace{-0.3in}\sum_{t=1}^{n}\left[H(Y_t,S_t|S_{[t-1]},Y_{[t-1]})-H(Y_t,S_t|\mathbf{W},S_{[t-1]},Y_{[t-1]})\right]\nonumber\\
&&\hspace{-0.3in}\overset{(i)}{=}\sum_{t=1}^{n}\left[H(Y_t|S_{[t]},Y_{[t-1]})-H(Y_t|\mathbf{W},S_{[t]},Y_{[t-1]})\right]\nonumber\\
&&\hspace{-0.3in}\overset{(ii)}{\leq}\sum_{t=1}^{n}\left[H(Y_t|S_{[t]})-H(Y_t|\mathbf{W},S_{[t]},Y_{[t-1]},\mathbf{T}_t)\right]\nonumber\\
&&\hspace{-0.3in}\overset{(iii)}{=}\sum_{t=1}^{n}\left[H(Y_t|S_{[t]})-H(Y_t|S_{[t]},\mathbf{T}_t)\right]\nonumber\\
&&\hspace{-0.262in}=\sum_{t=1}^{n}I(\mathbf{T}_t;Y_t|S_{[t]})\label{eq-c4}
\end{eqnarray}
where $(i)$ follows from the fact that $S_{t}$ is i.i.d. and independent of $\mathbf{W}$, in $(ii)$, $\mathbf{T}_t:=(T_t^a,T_t^b)$ are Shannon strategies whose realizations are mappings $t_t^i:S_t^{i}\rightarrow X_t^i$ for $i=\{a,b\}$ and thus $(ii)$ holds since conditioning reduces entropy. Finally, $(iii)$ follow since
\begin{eqnarray}
&&\hspace{-0.2in}P_{Y_t|\mathbf{W},S_t,S_{[t-1]},Y_{[t-1]},T_t^a,T_t^b}(y_t|\mathbf{w},s_t,s_{[t-1]},y_{[t-1]},t_t^a,t_t^b)\nonumber\\
&&=\sum_{s_t^a,s_t^b}P_{Y_t|S_t,S_t^a,S_t^b,T_t^a,T_t^b}(y_t|s_t,s_t^a,s_t^b,t_t^a,t_t^b)\nonumber\\ &&\hspace{0.9in}\times P_{S_t^a,S_t^b|S_t}(s_t^a,s_t^b|s_t)\nonumber\\
&&=P_{Y_t|S_t,T_t^a,T_t^b}(y_t|s_t,t_t^a,t_t^b)\label{eq-prob}
\end{eqnarray}
where the first equality is verified by (\ref{eq-ch}), where $x_t^i=t_t^i(s_t^i)$ for $i=\{a,b\}$, and by $\{S_t\}$ being i.i.d. and independent of $\mathbf{W}$. Now, let $\chi(\epsilon):=\frac{H(\epsilon)}{n(1-\epsilon)}$ and combining (\ref{eq-c3})-(\ref{eq-c4}) gives
\begin{eqnarray}
&&\hspace{-0.3in}R_a+R_b=\frac{1}{n}\log(|{\cal W}_a||{\cal W}_b|)\nonumber\\
&&\hspace{-0.15in}\leq\frac{1}{1-\epsilon}\frac{1}{n}\sum_{t=1}^nI(T_t^a,T_t^b;Y_t|S_{[t]})+\chi(\epsilon)+(n-1)\chi(\epsilon)\nonumber\\
&&\hspace{-0.15in}\overset{(a)}{\leq}\frac{1}{1-\epsilon}\frac{1}{n}\sum_{t=1}^nI(T_t^a,T_t^b;Y_t|S_{[t]})+\eta(\epsilon)\nonumber\\
&&\hspace{0.7in}-\frac{\epsilon}{1-\epsilon}\frac{1}{n}\sum_{t=1}^nI(T_t^a,T_t^b;Y_t|S_{[t]})\nonumber\\
&&\hspace{-0.15in}=\frac{1}{n}\sum_{t=1}^nI(T_t^a,T_t^b;Y_t|S_{[t]})+\eta(\epsilon)\label{eq-c5}
\end{eqnarray}
where $(a)$ is valid since $I(T_t^a,T_t^b;Y_t|S_{[t]})\leq \log|{\cal Y}|$. Furthermore,
\begin{eqnarray}
&&\hspace{-0.5in}I(T_t^a,T_t^b;Y_t|S_{[t]})\nonumber\\
&&=n\sum_{\mathbf{\sigma} \in {\cal S}^{(t-1)}}\alpha_{\mathbf{\sigma}}I(T_t^a,T_t^b;Y_t|S_t, S_{[t-1]}=\mathbf{\sigma}),
\end{eqnarray}
and substituting the above into (\ref{eq-c5}) yields (\ref{eq-lf1}).
\end{proof}
Observe now that, for any $t\geq1$, $I(T_t^a,T_t^b;Y_t|S_t,S_{[t-1]}=\mathbf{\sigma})$ is a function of the joint conditional distribution of channel state $S_t$, inputs $T_t^a, T_t^b$ and output $Y_t$ given the past realization $(S_{[t-1]}=\mathbf{\sigma})$. Hence, to complete the proof of the outer bound, we need to show that $P_{T_t^a,T_t^b,Y_t,S_t|S_{[t-1]}}(t_t^a,t_t^b,y_t,s_t|\mathbf{\sigma})$ factorizes as in (\ref{eq-joidist-f}). This is done in the lemma below. In particular, it is crucial to observe that the complete state observation at the decoder is enough to provide a product form on $T^a$ and $T^b$.
Before stating the lemma, let us introduce some more notations. Let $\mathbf{\sigma_a}$ and $\mathbf{\sigma_b}$ denote particular realizations of $S_{[t-1]}^a$ and $S_{[t-1]}^b$, respectively. Let
\begin{eqnarray}
&&\Upsilon_{\mathbf{\sigma_a}}^a(t^a):=\{w_a:\phi_{t}^{(a)}(w_a,s_{[t-1]}^a=\mathbf{\sigma_a})=t^a\}, \nonumber\\
&&\Upsilon_{\mathbf{\sigma_b}}^b(t^b):=\{w_b:\phi_{t}^{(b)}(w_b,s_{[t-1]}^b=\mathbf{\sigma_b})=t^b\}\label{eq-set1}
\end{eqnarray}
and
\begin{eqnarray}
\pi_{T^a}^{\mathbf{\sigma_a}}(t^a)&:=&\sum_{w_a\in \Upsilon_{\mathbf{\sigma_a}}^a(t^a)}\frac{1}{|{\cal W}_a|}, \nonumber\\
\pi_{T^b}^{\mathbf{\sigma_b}}(t^b)&:=&\sum_{w_b\in \Upsilon_{\mathbf{\sigma_b}}^b(t^b)}\frac{1}{|{\cal W}_b|}, \nonumber\\
\pi_{T^a}^{\mathbf{\sigma}}(t^a)&:=&\sum_{\mathbf{\sigma_a}}\pi_{T^a}^{\mathbf{\sigma_a}}(t^a)P_{S_{[t-1]}^a|S_{[t-1]}}(\mathbf{\sigma_a}|\mathbf{\sigma}),\nonumber\\ \pi_{T^b}^{\mathbf{\sigma}}(t^b)&:=&\sum_{\mathbf{\sigma_b}}\pi_{T^b}^{\mathbf{\sigma_b}}(t^b)P_{S_{[t-1]}^b|S_{[t-1]}}(\mathbf{\sigma_b}|\mathbf{\sigma}).\label{eq-pis}
\end{eqnarray}
\begin{lemma}\label{lem:fact-f}
For every $1\leq t\leq n$ and $\mathbf{\sigma} \in ({\cal S})^{t-1}$, the following holds
\begin{eqnarray}
&&\hspace{-0.45in}P_{T_t^a,T_t^b,Y_t,S_t|S_{[t-1]}}(t^a,t^b,y,s|\mathbf{\sigma})\nonumber\\
&&=P_S(s)P_{Y|S,T^a,T^b}(y|s,t^a,t^b)\pi_{T^a}^{\mathbf{\sigma}}(t^a)\pi_{T^b}^{\mathbf{\sigma}}(t^b).\label{eq-fact1-f}
\end{eqnarray}
\end{lemma}
\begin{proof}
Let $\textbf{s}:=(s,s_t^a,s_t^b)$ and observe that
\begin{eqnarray}
&&\hspace{-0.5in}P_{T_t^a,T_t^b,Y_t,S_t|S_{[t-1]}}(t^a,t^b,y,s|\mathbf{\sigma})\nonumber\\
&&=\sum_{s_t^a\in{\cal S}^a}\sum_{s_t^b\in{\cal S}^b}
P_{\textbf{S},T^a,T^b,Y|S_{[t-1]}}(\textbf{s},t^a,t^b,y|\mathbf{\sigma})\nonumber\\
&&=\sum_{s_t^a\in{\cal S}^a}\sum_{s_t^b\in{\cal S}^b}
P_{Y|\textbf{S},T^a,T^b}(y|\textbf{s},t^a,t^b)\nonumber\\
&&\hspace{0.8in} \times P_{\textbf{S},T^a,T^b|S_{[t-1]}}(\textbf{s},t^a,t^b|\mathbf{\sigma})\label{eq-fact2-f}
\end{eqnarray}
where the second equality verified by (\ref{eq-ch}) since $x_t^i=t_t^i(s_t^i)$ for $i=\{a,b\}$. Let us now consider the term $P_{\textbf{S},T^a,T^b|S_{[t-1]}}(\textbf{s},t^a,t^b|\mathbf{\sigma})$ above. We have the following
\begin{eqnarray}
&&\hspace{-0.5in}P_{\textbf{S},T^a,T^b|S_{[t-1]}}(\textbf{s},t^a,t^b|\mathbf{\sigma})\nonumber\\
&&\hspace{-0.4in}=\sum_{w_a\in {\cal W}_a}\sum_{w_b\in {\cal W}_b}\sum_{\mathbf{\sigma_a}}\sum_{\mathbf{\sigma_b}}\nonumber\\
&&\hspace{0.1in}P_{\mathbf{W},S_{[t-1]}^a,S_{[t-1]}^b,\textbf{S},T^a,T^b|S_{[t-1]}}(\mathbf{w},\mathbf{\sigma_a},\mathbf{\sigma_b},\textbf{s},t^a,t^b|\mathbf{\sigma})\nonumber\\
&&\hspace{-0.4in}\overset{(i)}{=}P_{\textbf{S}}(\textbf{s})\sum_{w_a\in {\cal W}_a}\sum_{w_b\in {\cal W}_b}\sum_{\mathbf{\sigma_a}}\sum_{\mathbf{\sigma_b}}\nonumber\\
&&\hspace{0.1in}P_{\mathbf{W},S_{[t-1]}^a,S_{[t-1]}^b,T^a,T^b|S_{[t-1]}}(\mathbf{w},\mathbf{\sigma_a},\mathbf{\sigma_b},t^a,t^b|\mathbf{\sigma})\nonumber\\
&&\hspace{-0.4in}\overset{(ii)}{=}P_{\textbf{S}}(\textbf{s})\sum_{w_a\in {\cal W}_a}\sum_{w_b\in {\cal W}_b}\sum_{\mathbf{\sigma_a}}\sum_{\mathbf{\sigma_b}}1_{\{t^l=\Phi_{t}^{(l)}(w_l,\mathbf{\sigma_l}),\enspace l=a,b\}}\nonumber\\
&&\hspace{0.1in}\times P_{\mathbf{W},S_{[t-1]}^a,S_{[t-1]}^b|S_{[t-1]}}(\mathbf{w},\mathbf{\sigma_a},\mathbf{\sigma_b}|\mathbf{\sigma})\nonumber\\
&&\hspace{-0.4in}\overset{(iii)}{=}P_{\textbf{S}}(\textbf{s})\sum_{w_a\in {\cal W}_a}\sum_{w_b\in {\cal W}_b}\sum_{\mathbf{\sigma_a}}\sum_{\mathbf{\sigma_b}}1_{\{t^l=\Phi_{t}^{(l)}(w_l,\mathbf{\sigma_l}),\enspace l=a,b\}}\nonumber\\
&&\hspace{0.1in}\times \frac{1}{|{\cal W}_a|}\frac{1}{|{\cal W}_b|}P_{S_{[t-1]}^a,S_{[t-1]}^b|S_{[t-1]}}(\mathbf{\sigma_a},\mathbf{\sigma_b}|\mathbf{\sigma})\nonumber\\
&&\hspace{-0.4in}\overset{(iv)}{=}P_{\textbf{S}}(\textbf{s})\sum_{\mathbf{\sigma_a}}P_{S_{[t-1]}^a|S_{[t-1]}}(\mathbf{\sigma_a}|\mathbf{\sigma})\sum_{\mathbf{\sigma_b}}P_{S_{[t-1]}^b|S_{[t-1]}}(\mathbf{\sigma_b}|\mathbf{\sigma})\nonumber\\
&&\hspace{-0.4in}\sum_{w_a\in {\cal W}_a} \frac{1}{|{\cal W}_a|}1_{\{t^a=\Phi_{t}^{(a)}(w_a,\mathbf{\sigma_a})\}}\sum_{w_b\in {\cal W}_b} \frac{1}{|{\cal W}_b|}1_{\{t^b=\Phi_{t}^{(b)}(w_b,\mathbf{\sigma_b})\}}\nonumber\\
&&\hspace{-0.4in}\overset{(v)}{=} P_{\textbf{S}}(\textbf{s})\sum_{\mathbf{\sigma_a}}P_{S_{[t-1]}^a|S_{[t-1]}}(\mathbf{\sigma_a}|\mathbf{\sigma}) \sum_{w_a\in \Upsilon_{\mathbf{\sigma_a}}^a(t^a)}\frac{1}{|{\cal W}_a|} \nonumber\\
&&\hspace{0.13in} \sum_{\mathbf{\sigma_b}}P_{S_{[t-1]}^b|S_{[t-1]}}(\mathbf{\sigma_b}|\mathbf{\sigma})\sum_{w_b\in \Upsilon_{\mathbf{\sigma_b}}^b(t^b)}\frac{1}{|{\cal W}_b|}\nonumber\\
&&\hspace{-0.4in}\overset{(vi)}{=} P_{\textbf{S}}(\textbf{s})\sum_{\mathbf{\sigma_a}}P_{S_{[t-1]}^a|S_{[t-1]}}(\mathbf{\sigma_a}|\mathbf{\sigma})\pi_{T^a}^{\mathbf{\sigma_a}}(t^a)\nonumber\\ &&\hspace{0.15in}\sum_{\mathbf{\sigma_b}}P_{S_{[t-1]}^b|S_{[t-1]}}(\mathbf{\sigma_b}|\mathbf{\sigma})\pi_{T^b}^{\mathbf{\sigma_b}}(t^b)\nonumber\\
&&\hspace{-0.4in}\overset{(vii)}{=}P_{\textbf{S}}(\textbf{s})\pi_{T^a}^{\mathbf{\sigma}}(t^a)\pi_{T^b}^{\mathbf{\sigma}}(t^b)\label{eq-pro1-f}
\end{eqnarray}
where $(i)$ is valid since the current state is independent of $\mathbf{W}$ and $(T^a,T^b)$, $(ii)$ is valid by (\ref{eq-t2}), $(iii)$ is valid since $\mathbf{W}$ is independent from the state processes, $(iv)$ is valid by (\ref{eq-sta-no}) and (\ref{eq-t2}), $(v)$ is valid due to (\ref{eq-set1}) and $(vi)-(vii)$ is valid due to (\ref{eq-pis}). Substituting (\ref{eq-pro1-f}) into (\ref{eq-fact2-f}) proves the lemma.
\end{proof}

We can now complete the proof of Theorem \ref{the-main-outer-f}. With Lemma \ref{lem:conv-f} it is shown that the sum of any achievable rate pair can be approximated by the convex combinations of rate conditions given in (\ref{eq-ra3-f}) which are indexed by $\mathbf{\sigma} \in {\cal S}^{(n)}$ and satisfy (\ref{eq-joidist-f}) for joint state-input-output distributions.
More explicitly, we have
\begin{eqnarray}
R_a+R_b&\leq& \sum_{\mathbf{\sigma} \in {\cal S}^{(n)}}\alpha_{\mathbf{\sigma}} I(T_t^a,T_t^b;Y_t|S_t,S_{[t-1]}=\mathbf{\sigma})+\eta(\epsilon)\nonumber\\
&=&\sum_{\mathbf{\sigma} \in {\cal S}^{(n)}}\alpha_{\mathbf{\sigma}} I(T_t^a,T_t^b;Y_t|S_t)_{\pi_{T^a}^{\mathbf{\sigma}}(t^a)\pi_{T^b}^{\mathbf{\sigma}}(t^b)}+\eta(\epsilon)\nonumber\\
&\leq&\sup_{\left(\pi_{T^a}^{\mathbf{\sigma}}(t^a)\pi_{T^b}^{\mathbf{\sigma}}(t^b),\enspace\mathbf{\sigma}\right)}I(T_t^a,T_t^b;Y_t|S_t)+\eta(\epsilon)\nonumber\\
&\leq&\sup_{\left(\pi_{T^a}(t^a)\pi_{T^b}(t^b)\in \Pi\right)}I(T_t^a,T_t^b;Y_t|S_t)+\eta(\epsilon)\nonumber
\end{eqnarray}
where the second step is valid since $I(T_t^a,T_t^b;Y_t|S_t,S_{[t-1]}=\mathbf{\sigma})$ is a function of the joint conditional distribution of channel state $S_t$, inputs $T_t^a, T_t^b$ and output $Y_t$ given the past realization $(S_{[t-1]}=\mathbf{\sigma})$. Hence, since $\lim_{n \rightarrow \infty}\eta(\epsilon)=0$, any achievable pair satisfies $R_a+R_b \leq \sup_{\pi_{T^a}(t^a)\pi_{T^b}(t^b)}I(T^a,T^b;Y|S)$.
\end{proof}
As a direct consequence of Theorem \ref{the-main-outer-f}, we have the following corollary.
\begin{corollary}\label{cor2}
\begin{eqnarray}
{\cal C}^{FS}_{\sum}=\sup_{\pi_{T^a}(t^a)\pi_{T^b}(t^b)}I(T^a,T^b;Y|S)\nonumber.
\end{eqnarray}
\end{corollary}
\begin{remark}
One main observation about the proof of Theorem \ref{the-main-outer-f} is the fact that, once we have the complete state information, conditioning on which allows a product form on $T^a$ and $T^b$, there is no loss of optimality (for the sum-rate capacity) in using associated memoryless team policies instead of using all the past information at the receiver. This fact is observed in \cite{GiacomoSerdar} when the information at the encoders are asymmetric quantized version of the information at the decoder.
\end{remark}
\begin{remark}
It should be noted that the main difference between the problem that we consider here and the one considered in \cite{GiacomoSerdar} is the information at the decoder about the information at the encoders. More explicitly, in \cite{GiacomoSerdar}, the information at the encoders are available at the decoder and as such, as the authors explicitly mention in their paper, the decoder does not need to estimate the coding policies used in a decentralized time-sharing. From this perspective, the main contribution of our work can be thought as showing that when this is not the case, by enlarging the input space, there is no loss of optimality (for the sum-rate capacity) if the optimization is performed by ignoring the past information at the encoders given that the decoder has complete CSI.
\end{remark}
\section{Conclusion and Remarks}\label{conc}
The present paper has investigated the memoryless FS-MAC with asymmetric noisy CSI at the encoders and complete CSI at the decoder. Single letter inner and outer bounds  are presented when the channel state is a sequence of i.i.d. random variables. The main contribution of the paper, i.e., the tight converse for the sum-rate capacity and hence an outer bound to the capacity region, is realized by observing that the information available at the decoder is enough to attain a product form on the channel input functions and hence there is no loss of optimality if we ignore the past noisy state information at the encoders.

\end{document}